\documentclass[copyright,creativecommons]{eptcs}
\providecommand{\event}{WRS09} 
\providecommand{\volume}{??}
\providecommand{\anno}{20??}
\providecommand{\firstpage}{1}
\providecommand{\eid}{??}

\usepackage{proof}
\usepackage{amsmath,amsfonts,amssymb,times,xspace}
\usepackage{stmaryrd,pifont,amsfonts,amssymb}
\usepackage{graphicx,fancyhdr}
\usepackage[T1]{fontenc}
\usepackage{verbatim}
\usepackage{url}
\usepackage[all,tips]{xypic}
\usepackage{color}

\newcommand{\fzeta} {\ensuremath{\zeta^{<\omega}}\xspace}
\newcommand{\nzeta} {\ensuremath{\lambda}\xspace}
\newcommand{\ezeta} {\ensuremath{\zeta_{\nzeta}}}
\newcommand{\szeta} {\ensuremath{\zeta_{\nzeta}^{<\omega}}}

\newcommand{\ext}[1] {\ensuremath{\zeta_{#1}}\xspace}
\newcommand{\sext}[1] {\ensuremath{\zeta_{#1}^{<\omega}}\xspace}
\newcommand{\wrt} {\textit{w.r.t.}\xspace}

\newcommand{\reductions} {\ensuremath{\Gamma}}

\newcommand{\derivations}[1] {\allDeriv{#1}}

\newcommand{\allDeriv}[1] {\ensuremath{\Gamma_{#1}^+}}
\newcommand{\allDerivStar}[1] {\ensuremath{\Gamma^{*}_{#1}}}
\newcommand{\allDerivInf}[1] {\ensuremath{\Gamma^{\omega}_{#1}}}
\newcommand{\elan}    {\textsf{ELAN}\xspace}
\newcommand{\maude}   {\textsf{Maude}\xspace}
\newcommand{\stratego}    {\textsf{Stratego}\xspace}
\newcommand{\strafunski}    {\textsf{Strafunski}\xspace}
\newcommand{\tom}     {\textsf{TOM}\xspace}

\newcommand{\LCF}     {\textsf{LCF}\xspace}
\newcommand{\app}       {\mathrel{\!\!}}

\newcommand{\ars}{abstract reduction system\xspace}

\newcommand{\aos}{\ensuremath{\ca A=(\ca O,\ca L,\reductions)}}
\newcommand{\flechup}[1]{\xrightarrow{#1}}

\newcommand{\caA}{{\cal A}}

\newcommand{\implique}{\ensuremath{\Rightarrow}}

\newcommand{\cx}[1]{\ensuremath{\left[#1\right]\vspace{1pt}}}

\newcommand{\Ee}{{\cal O}}

\newcommand{\vide}{\Lambda}

\newcommand{\cobj}{\ensuremath{\Ee^{[\ca A]}}}

\newcommand{\ca}[1]{\ensuremath{\mathcal{#1}}}
\renewcommand{\emptyset}{\ensuremath{\varnothing}}
\newcommand{\sema}[1]{\ensuremath{\left\llbracket #1 \right\rrbracket}}

\ifx\prop\undefined
\newtheorem{definition}{Definition}
\fi
\ifx\example\undefined
\newtheorem{example}{Example}
\fi
\ifx\proposition\undefined
\newtheorem{proposition}{Proposition}
\fi
\ifx\theorem\undefined
\newtheorem{theorem}{Theorem}
\fi
\ifx\lemma\undefined
\newtheorem{lemma}{Lemma}
\fi
\ifx\corollary\undefined
\newtheorem{corollary}{Corollary}
\fi
\ifx\proof\undefined
\newcommand{\proof}[1]{\begin{description}
            \item[Proof:] #1 \hfill\(\Box\) \end{description}}
\fi

\newcommand{\carone}{\ensuremath{\mathsf{car}_1}}
\newcommand{\cartwo}{\ensuremath{\mathsf{car}_2}}
\newcommand{\signalone}{\ensuremath{\mathsf{signal}_1}}
\newcommand{\signaltwo}{\ensuremath{\mathsf{signal}_2}}
\newcommand{\crossone}{\ensuremath{\mathsf{cross}_1}}
\newcommand{\crosstwo}{\ensuremath{\mathsf{cross}_2}}
\newcommand{\sem}[1]{\ensuremath{\zeta_{#1}}} 
\begin{document}

\title{Extensional and Intensional Strategies}

\author{Tony Bourdier
\hfill Horatiu Cirstea
\institute{INRIA Nancy Grand-Est
\hfill Nancy Universit\'{e} \hfill LORIA \\ BP 239, 
54506 Vandoeuvre-les-Nancy Cedex, France }
\and Daniel J.\  Dougherty
\institute{
\and Worcester Polytechnic Institute\\ Worcester, MA, 01609, USA }
\and H\'el\`ene Kirchner
\institute{
\and INRIA Bordeaux Sud-Ouest\\
351, Cours de la Lib\'eration, 33405 Talence, France}
}

\def\titlerunning{Extensional and Intensional Strategies}
\def\authorrunning{T. Bourdier, H. Cirstea, D. Dougherty \& H. Kirchner}

\maketitle

\begin{abstract}

  This paper is a contribution to the theoretical
  foundations of strategies.  We first present a general definition of
  abstract strategies which is {\em extensional} in the sense that a
  strategy is defined explicitly as a set of derivations of an abstract
  reduction system.  We then move to a more {\em intensional} definition
  supporting the abstract view but more operational in the sense that it
  describes a means for determining such a set.  
  We characterize the class of extensional strategies
  that can be defined intensionally. We also give some hints
  towards a logical characterization of intensional strategies and
  propose a few challenging perspectives.

\end{abstract}

\section{Introduction}

Rule-based reasoning is present in many domains of computer science:
in formal specifications, rewriting allows prototyping
specifications; in theorem proving, rewrite rules are used for dealing with equality,
simplifying the formulas and pruning the search space; in programming
languages, rules can be explicit like in PROLOG, OBJ or ML, or hidden
in the operational semantics; expert systems use rules to describe
actions to perform; in constraint logic programming, solvers are
described via rules transforming constraint systems.  XML document
transformations, access-control policies or bio-chemical reactions are
a few examples of application domains.

Nevertheless, deterministic rule-based computations or deductions are
often not sufficient to capture every computation or proof
development. A formal mechanism is needed, for instance, to
sequentialize the search for different solutions, to check context
conditions, to request user input to instantiate variables, to process
sub-goals in a particular order, etc. This is the place where the
notion of strategy comes in.

Strategies have been introduced in functional programming (Lisp, ML,
Haskell, OBJ), logic programming (PROLOG, CHR), logic-functional
languages (Curry, Toy) and constraint programming (CLP).  Reduction
strategies in term rewriting study which expressions should be
selected for evaluation and which rules should be applied. These
choices usually increase efficiency of evaluation but may affect
fundamental properties of computations such as confluence or
(non-)termination.  Programming languages like \elan, \maude and
\stratego allow for the explicit definition of the evaluation
strategy, whereas languages like Clean, Curry, and Haskell allow for
its modification.
In theorem proving environments, including automated theorem provers,
proof checkers, and logical frameworks, strategies (also called
tactics or tacticals in some contexts) are used for various purposes,
such as proof search and proof planning, restriction of search spaces,
specification of control components, combination of different proof
techniques and computation paradigms, or meta-level programming in
reasoning systems.
Strategies are increasingly useful as a component in systems with
computation and deduction explicitly
interacting~\cite{TPM-DHK-JAR-03,EMOMV07}.  The complementarity
between deduction and computation, as emphasized in particular in
deduction modulo~\cite{TPM-DHK-JAR-03}, allows us to now envision a
completely new generation of proof assistants where customized
deductions are performed modulo appropriate and user definable
computations~\cite{Brauner:2007fk,Brauner:2007uq}.

In the fields of system design and verification,
 \emph{games} ---most often two-person path-forming games over graphs---
 have emerged as a key tool.  Such games have been studied since the
 first half of 20th century in descriptive set theory \cite{Kechris95},
 and 
 they have been adapted and generalized for applications in formal
 verification; introductions can be found in
 \cite{dagstuhl2001,BerwangerGK03,Walukiewicz04}.  Related applications
 appear in logic \cite{Blass92}, %
 planning \cite{PistoreV03}, %
 and synthesis \cite{ArnoldVW03}.
At first glance the coincidence of the term ``strategy'' in the domains
of rewriting and games appears to be no more than a pun.  But it turns
out to be surprisingly fruitful to explore the connection and to be
guided in the study of the foundations of strategies by some of the
insights in the literature of games.  This point of view is further developed 
in this paper.

In order to illustrate the variety of situations involving the notion
of strategy, let us give a few examples.
In~\cite{KlopOostromRaamsdonk}, the authors describe a
non-deterministic strategy for higher-order rewriting: it amounts to
choose an outermost redex and skip redexes that do no contribute to
the normal form because they are in a cycle.  In proof assistants like
Coq \cite{Del00}, tactics are used to describe proof search
strategies. For instance, the {\tt orelse} tactic in \LCF is defined
as follows: given two tactics A and B, apply tactic B only if the
application of tactic A either failed or did not modify the proof.  In
constraint solving, intricate strategies have been defined by
combining choice points setting, forward or backward checking,
enumeration strategy of values in finite domains, and selection of
solutions. Examples can be found in~\cite{Castro:FI:1998}.  In game
theory, the notion of strategy is crucial to determine the next move
for each player \cite{ArnoldVW03}. In~\cite{Dougherty08}, the idea is
applied to computation of normal forms: two players $W$ and $B$ with
respective rules $R_W$ and $R_B$ play a game by rewriting terms in
the combined signature and we want to know if there exists a winning
strategy to reach the normal form.

Strategies are thus ubiquitous in automated deduction and reasoning
systems, yet only recently have they been studied in their own
right. This paper is a contribution to the concept definition and its
theoretical foundation. We try to reconcile different points of view
and to compare their expressive power.

In Section~\ref{ARS}, we recall the definitions related to abstract
reduction systems, before giving in Section~\ref{abstractStrategies}
the definition of an abstract strategy as a subset of reduction
sequences called derivations.  
In Section~\ref{se:IntensionalStrategies}, we give an intensional definition
of strategies compliant with the abstract view but more operational in
the sense that it describes a means of building a subset of
derivations by associating to a reduction-in-progress the possible 
next steps in the reduction sequence. Then intensional strategies with
memory are defined. This gives the expressive power to build next step
with the knowledge of past steps in the derivation. 
Section~\ref{se:express} explores which abstract strategies can be actually expressed by intensional ones.
In order to increase the expressive power of intensional strategies,
we eventually propose in Section~\ref{sec:log-int-strat} to define intensional strategies with 
an accepting condition.
Further research questions are presented in Section~\ref{Conclusion}.

\section{Abstract reduction systems}
\label{ARS}

When abstracting the notion of strategy, one important preliminary
remark is that we need to start from an appropriate notion of
\textit{abstract reduction system} (ARS) based on the notion of
oriented labeled graph instead of binary relation. This is due to the
fact that, speaking of derivations, we need to make a difference
between ``being in relation'' and ``being connected''. Typically
modeling ARS as relations as in~\cite{BaaderNipkowREW-98} allows us to
say that, \textit{e.g.}, $a$ and $b$ are in relation but not that
there may be several different ways to derive $b$ from
$a$. Consequently, we need to use a more expressive approach, similar
to the one proposed in~\cite{TereseBook2003,TereseChapter8-2003} based on a
notion of oriented graph.  Our definition is similar to the one given
in~\cite{KKK08} with the slight difference that we make more precise
the definition of steps and labels.
Similarly to the step-based
definition of an \ars\ of~\cite{TereseBook2003}, this definition that
identifies the reduction steps avoids the so-called \emph{syntactic
accidents}~\cite{Levy78},  related to different but indistinguishable
derivations. 

\begin{definition}[Abstract reduction system]
\label{def:ARS}
Given a countable set of objects $\ca O$ and a countable set of labels
$\ca L$ mutually disjoint, an \emph{\ars} (ARS) is a triple $(\ca
O,\ca L,\reductions)$ such that $\reductions$ is a functional relation from
$\ca O \times \ca L$ to $\ca O$: formally, $\reductions \subseteq \ca O
\times \ca L\times \ca O$ and $(a,\phi,b_1) \in \reductions$ and
$(a,\phi,b_2)\in \reductions$ implies $b_1=b_2$.

The tuples $(a,\phi,b) \in \reductions$ are called \emph{steps} and are
often denoted by $a \flechup{\phi} b$. We say that $a$ is the
\emph{source} of $a \flechup{\phi} b$, $b$ its \emph{target} and
$\phi$ its \emph{label}. Moreover, two steps are \emph{composable} if the target
of the former is the source of the latter.

\end{definition}

Like for graphs and labeled transition systems, in order to support
intuition, we often use the obvious graphical representation to
denote the corresponding ARS.

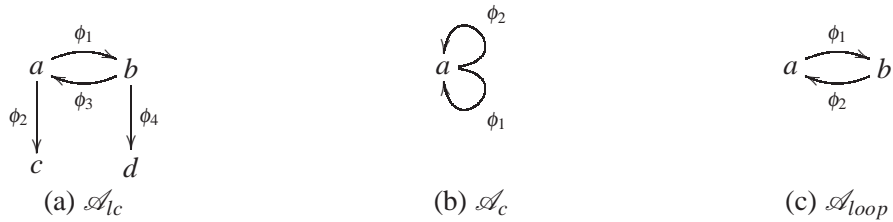
\begin{figure}[!th]
\begin{center}
\begin{tabular}{c@{~~~~~~~~~~~~~~~~~~~~~~~~~~~~~~~~~~~~}c@{~~~~~~~~~~~~~~~~~~~~~~~~~~~~~~~~~~~~}c}
\xymatrix{
a  \ar@/^/[r]^{\phi_1} \ar[d]_{\phi_2} &  b \ar@/^/[l]^{\phi_3} \ar[d]^{\phi_4} \\
c        & d
}
&
\xymatrix{ a \ar@(r,d)[]^{\phi_1}\ar@(r,u)[]_{\phi_2}}
&
\xymatrix{
a  \ar@/^/[r]^{\phi_1} &  b \ar@/^/[l]^{\phi_2} 
}
\\
(a) ${\cal A}_{lc}$  & (b)  ${\cal A}_{c}$  & (c) ${\cal A}_{loop}$
\end{tabular}
\end{center}
\caption{Graphical representation of abstract reduction systems}
\label{fig:exARS}
\end{figure}

\begin{example}[Abstract reduction systems] \label{ex:simples}
The \ars\ 
$${\cal
  A}_{lc}=(\{a,b,c,d\},\{\phi_1,\phi_2,\phi_3,\phi_4\},\{(a,\phi_1,b),(a,\phi_2,c),(b,\phi_3,a),(b,\phi_4,d)\})$$
is depicted in Figure~\ref{fig:exARS}(a).

The interest of using the above definition of abstract reduction
systems instead of the classical one based on binary relations is
illustrated by the \ars\ 
$${\cal A}_c=(\{a\},\{\phi_1,\phi_2\},\{(a,\phi_1,a),(a,\phi_2,a)\})$$
depicted in Figure~\ref{fig:exARS}(b).

The ``looping'' \ars\ 
$${\cal A}_{loop}=(\{a,b\},\{\phi_1,\phi_2\},\{(a,\phi_1,b),(b,\phi_2,a)\})$$
is depicted in Figure~\ref{fig:exARS}(c).

The \ars
$${\cal A}_{nat}=(\{a_i ~|~ i \in \mathbb N \}, \{\phi_i ~|~  i \in \mathbb N \},\{(a_i,\phi_i,a_{i+1}) ~|~  i \in \mathbb N\})$$
 has   infinite (but countable) sets of objects and labels, and its
 relation can be  depicted by: 
$$a_0 \flechup{\phi_0} a_1   \flechup{\phi_1} \ldots a_i \flechup{\phi_i} a_{i+1}  \ldots$$

Another  \ars\ with an infinite (but countable) set of derivations
{starting from a same source} is
$${\cal A}_{ex}=(\{a_i^j ~|~  i,j \in \mathbb N \}, \{\phi_i^j ~|~ 0\leq i <j \},
\{(a_0^0,\phi_0^j,a_1^j) ~|~  1\leq j \} \cup \{ (a_i^j,\phi_i^j,a_{i+1}^j) ~|~   1\leq i <j \})$$
{whose relation can be depicted by:}\\
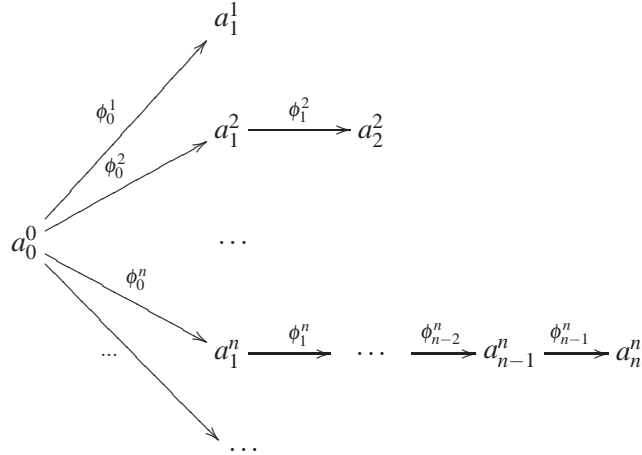
$$
 \xymatrix{ 
 		 		&& a_1^1	\\
 		 		&& a_1^2	\ar[r]^{\phi_1^2} & a_2^2	\\
 a^0_0\ar[uurr]^{\phi_0^1}
 			\ar[urr]^{\phi_0^2}
			\ar[drr]^{\phi_0^n}
 			\ar[ddrr]_{\ldots}	
 				&&	~\ldots\\
 				&& a_1^n \ar[r]^{\phi_1^n}  & ~~\ldots~~ \ar[r]^{\phi_{n-2}^n} & a_{n-1}^n \ar[r]^{\phi_{n-1}^n}& a_{n}^n\\
 				&& ~~~~\ldots
 	}
 $$
\end{example}

The condition that $\reductions$ is a functional relation implies that
an ARS is a particular case of a labeled transition system. As we will
see in what follows, labels characterize the way an object is
transformed: given an object and a
transformation, there is at most one object resulting from
the transformation applied to this particular object.  So Definition~\ref{def:ARS} does not authorize
for instance to have $\phi_1=\phi_2=\phi$ in ${\cal A}_{lc}$
of Example~\ref{fig:exARS}.

The next definitions can be seen as a renaming of usual ones in graph
theory. Their interest is to allow us to define uniformly derivations
and strategies in different contexts.

\begin{definition}[Derivation]
  Given an \ars\ $\ca A=(\ca O,\ca L,\reductions)$ we call
  \emph{derivation} over $\ca A$ any sequence $\pi$ of steps
  $\big((t_i,\phi_i,t_{i+1})\big)_{i\in \Im}$ for any right-open
  interval $\Im\subseteq \mathbb N$ starting from $0$.  If $\Im$
  contains at least one element, then:
  \begin{itemize}
  \item $Dom(\pi)=t_0$ is called the \emph{source} (or domain) of
    $\pi$,
  \item $l(\pi) = (\phi_i)_{i\in \Im}$ is a sequence called
    \emph{label} of $\pi$,
  \item For any non empty (possibly right-open) subinterval $\Im' \subseteq \Im$,
    $\pi'=\big((t_i,\phi_i,t_{i+1})\big)_{i\in \Im'}$ is a factor of
    $\pi$. If $\Im'$ contains $0$, then $\pi'$ is a prefix of $\pi$.
  If $\Im'\neq \Im$, $\pi'$ is a strict factor (or prefix) of $\pi$.
\end{itemize}
If $\Im$ is finite, it has a smallest upper bound denoted by
$n_\Im$ or simply $n$ and then:
\begin{itemize}
\item $Im(\pi) = t_n$ is called the \emph{target} (or image) of $\pi$,
\item $|\,\pi\,| = card(\Im)$ is called the \emph{length} of $\pi$,
\end{itemize}
In such a case, $\pi$ is said to be finite and is also denoted by
$\pi=(t_0,l(\pi),t_n)$ or $t_0\flechup{l(\pi)}t_n$.  The
sequence containing no step is called \emph{empty derivation} and is
denoted by $\Lambda_\Gamma$ and by convention
$l(\Lambda_\Gamma)=\epsilon$ where $\epsilon$ is the empty sequence of 
elements of $\ca L$.
\end{definition}

Note that a step may be considered as a derivation of length 1.

We denote by \allDerivInf{\ca A} (resp. \allDerivStar{\ca A},
resp. \allDeriv{\ca A}) the set of all derivations (resp. finite,
resp. non-empty and finite) over $\ca A$.

\begin{definition}[Composable derivations]
  Two derivations of the form $\pi_1=\big((t_i,\phi_i,t_{i+1})\big)_{i\in \Im_1}$
  and $\pi_2=\big((u_i,\phi_i,u_{i+1})\big)_{i\in \Im_2}$ over a same
  \ars\ $\ca A=(\ca O,\ca L,\reductions)$ are \emph{composable}
  iff either one of the derivations is empty or $\Im_1$ is finite and
  then $t_{n_1} = u_{0}$ where $n_1$ is the smallest upper bound
  of $\Im_1$. In such a case, the composition of $\pi_1$ and
  $\pi_2$ is the unique derivation $\pi = \big( (v_i,\phi_i,v_{i+1})
  \big)_{i\in\Im}$ denoted by $\pi=\pi_1\pi_2$ such that for all $j <
  |\,\pi_1\,|$, $v_{j} = t_{j}$ and for all $j\geq |\,\pi_1\,|$,
  $v_{j} = u_{j-|\,\pi_1\,|}$.
\end{definition}

The composition is associative and has a neutral element which is
$\Lambda_\Gamma$. 
Adopting the product notation, we denote $\prod_{i=1}^n \pi_i =
\pi_1 \ldots \pi_n$, $\pi^n = \prod_{i=1}^n \pi$ and $\pi^\omega =
\prod_{i\in\mathbb N} \pi$.

\begin{example}[Derivations] Following the previous examples, we have:
  \begin{enumerate}
  \item $\allDerivInf{{\ca A}_{lc}}$ contains for instance $\pi_1, \pi_1\pi_3,
    \pi_1\pi_4,\pi_1\pi_3\pi_1, (\pi_1\pi_3)^n, (\pi_1\pi_3)^\omega,
    \ldots$, with $\pi_1 = (a,\phi_1,b)$, $\pi_2 = (a,\phi_2,c)$,
    $\pi_3 = (b,\phi_3,a)$, $\pi_4 = (b,\phi_4,d)$;
  \item $\allDeriv{{\ca A}_{c}} = \left\{ \prod_{i=1}^k \pi_1^{n_i}
    \pi_2^{m_i}~|~k\geq 1,(n_i,m_i)\in\mathbb N^2,(n_i+m_i) > 0\right\}$
    with $\pi_1 = (a,\phi_1,a)$ and $\pi_2 = (a,\phi_2,a)$.
  \item $\allDeriv{{\ca A}_{loop}} = \left\{ (\pi_1\pi_2)^n\pi_1,
    (\pi_2\pi_1)^m\pi_2, (\pi_1\pi_2)^{n+1}, (\pi_2\pi_1)^{m+1} ~|~ (n,m) \in \mathbb N^2\right\}$
    with $\pi_1 = (a,\phi_1,b)$ and $\pi_2 = (b,\phi_2,a)$.
$\allDerivInf{{\ca A}_{loop}}$ contains also $(\pi_1\pi_2)^{\omega}$ and  $(\pi_2\pi_1)^{\omega}$.

  \end{enumerate}
\end{example}

\section{Abstract strategies}
\label{abstractStrategies}

Several different definitions of the notion of strategy have been
given in the  literature. Here is a sampling.

\begin{itemize}
\item
A strategy is a map $F$ from terms to terms such that $t \mapsto F(t)$
(\cite{Barendregt84} in the context of the $\lambda$-calculus,
\cite{FernandezMS05} in the context of explicit substitutions calculi).
\item
A strategy is a sub ARS having the same set of normal forms
(\cite{TereseBook2003,TereseStrategies2003} in the context of abstract reduction systems).
\item
A strategy is a plan for achieving a complex transformation using a
set of rules (\cite{VisserJSC05} in the context of program
transformations).
\item 
A strategy is a set of proof terms in rewriting logic
(\cite{BKKR-IJFCS-2001} in the \elan\ system).
\item 
A strategy is a (higher-order) function (in
\elan~\cite{borovansky02a}, \maude~\cite{MartiOlietMeseguerVerdejo04})
that can apply to other strategies.
\item 
A strategy is a $\rho$-term in the $\rho$-calculus
\cite{rhoCalIGLP-I+II-2001}.
 \item
A strategy is a subset of the set of all rewriting
derivations~\cite{KKK08}, in the context of abstract strategies for
deduction and computation.  This view is further detailed below.
\item
A strategy is a partial function that associates to a
reduction-in-progress, the possible next steps in the reduction
sequence.  
Here, the strategy as a function depends only on the object and the
derivation so far.  This notion of strategy coincides with the
definition of strategy in sequential path-building games,
with applications to planning, verification and
synthesis of concurrent systems~\cite{Dougherty08}.
\end{itemize}

Among these various definitions some of them are {\em extensional} in
the sense that a strategy is taken explicitly as a set of derivations,
while others are {\em intensional} in the sense that they describe a
means for determining such a set.  We focus in this paper on the two
last definitions and explore their implications and their relations.

We use a general definition
slightly different from the one used in~\cite{TereseStrategies2003}.
This approach has already been proposed in~\cite{KKK08} and here we essentially 
improve and detail a few definitions.

\begin{definition}[Abstract Strategy] \label{def:abstractStrategy}
Given an ARS $\ca A$, an \emph{abstract strategy} $\zeta$ over $\ca A$ is a
subset of non-empty derivations of  \allDerivInf{\ca A}.
\end{definition}
A strategy can be a finite or an infinite set of derivations, and the
derivations themselves can be finite or infinite in length.

Let us introduce some terminology that  will be useful later on.

\begin{definition}[Factor-closed, prefix-closed, closed by composition] 
Given an abstract strategy $\zeta$ over an ARS $\ca A$,
\begin{itemize}
\item $\zeta$ is \emph{factor-closed} (resp. \emph{prefix-closed}) iff
  for any derivation $\pi$ in $\zeta$, any factor (resp. prefix) of
  $\pi$ is also in $\zeta$.
\item $\zeta$ is \emph{closed by composition} iff for any two
  composable derivations $\pi, \pi'$ in $\zeta$, their composition
  $\pi \pi'$ is in $\zeta$ too.
 \end{itemize}
\end{definition}

An abstract strategy over an \ars  $\ca A = (\ca O,\ca L,\reductions)$
induces a (partial) function from $\ca O$ to $2^{\ca O}$.
This functional point of view has been already proposed
in~\cite{BKKR-IJFCS-2001}; we just briefly recall it in our formalism.

The \emph{domain} of a strategy $\zeta$ is the set of objects that are
source of a derivation in $\zeta$:
    $$Dom(\zeta) = \bigcup_{\pi \in \zeta} Dom(\pi)$$
The application of a strategy is defined (only) on the objects of its
domain. The application of a strategy $\zeta$ on $a\in Dom(\zeta)$ is
denoted $\app{\zeta}{a}$ and is defined as the set of all objects that
can be reached from $a$ using a finite derivation in $\zeta$:
    $$\app{\zeta}{a}
    = \{Im(\pi) ~|~ \pi \in \zeta, \pi \mbox{ finite and } Dom(\pi)=a \}$$
    If $a \not\in Dom(\zeta)$ we say that $\zeta$ \emph{fails} on $a$
   ($\zeta$ contains no derivation of source $a$).

    If $a \in Dom(\zeta)$ and $\app{\zeta}{a}=\emptyset$, we say that
    the strategy $\zeta$ is \emph{indeterminate} on $a$.  In fact,
    $\zeta$ is indeterminate on $a$ if and only if $a \in Dom(\zeta)$
    and $\zeta$ contains no
    finite derivation starting from $a$.

\begin{example}[Strategies]\label{ex:strategies}
Let us consider the \ars\ ${\cal A}_{lc}$ of Example~\ref{ex:simples} and
define the following strategies:
 \begin{enumerate}
  \item The strategy $\zeta_{u}$ $= \allDerivInf{{\ca A}_{lc}}$, also
    called the \emph{Universal} strategy~\cite{KKK08} (\wrt\ $\ca A_{lc}$), contains
    all the derivations of $\ca A_{lc}$. We have $\app{\zeta_{u}}{a} =\,\,
    \app{\zeta_{u}}{b} = \{a,b,c,d\}$ and $\zeta_{u}$ fails on $c$ and
    $d$.
 \item The strategy $\zeta_{f} = \emptyset$, also called \emph{Fail},
   contains no derivation and thus fails on any $x\in\{a,b,c,d\}$.
  \item No matter which derivation of the strategy $\zeta_{c} = \left\{\left(a\flechup{\phi_1
    \phi_3} a\right)^n  a \flechup{\phi_2} c~|~n\geq 0\right\}$
    is considered, the object $a$ eventually reduces to $c$: $\app{\zeta_{c}}{a}
    = \{c\}$. $\zeta_{c}$ fails on $b$, $c$ and $d$.
  \item The strategy $\zeta_{\omega} = {\left\{\left(a\flechup{\phi_1
      \phi_3} a\right)^\omega \right\}}$ is indeterminate on $a$ and
    fails on $b$, $c$ and $d$.
  \end{enumerate}
 The strategies $\zeta_{u}$ and $\zeta_{f}$ are prefix closed while
 $\zeta_{c}$ and $\zeta_{\omega}$ are not.\\
The so-called \emph{Universal} and \emph{Fail} strategies introduced
in Example~\ref{ex:strategies} can be obviously defined over any
\ars. 
\end{example}

There is a natural topology on the space of derivations in an
\ars. The set of (finite and infinite) derivations is essentially the
\emph{Kahn domain}\cite{Escardo04} over the set of labels.
 The basic open sets in this topology
are the ``intervals'', the sets $B_{\pi'} = \{\pi \mid \pi' \mbox{ is a
  prefix of } \pi\}$ as $\pi'$ ranges over the finite derivations.
Under this topology we have the following characterization of the
\emph{closed} sets.

\begin{definition}\label{closed}
If $\zeta$ is a set of derivations, a \emph{limit point} of $\zeta$ is a
derivation $\pi$ with the property that 
every finite prefix  $\pi_0$  of $\pi$ is a finite prefix of some derivation in $\zeta$.
 A set $\zeta$ of (finite or infinite) derivations is \emph{closed} if
  it contains all of its limit points.
 Equivalently: for every derivation $\pi$
  not in $\zeta$, there is a finite prefix
$\pi_0$ of $\pi$ such that every extension of 
$\pi_0$ fails to be in $\zeta$. 
\end{definition}

As observed by Alpern and Schneider~\cite{alpern87recognizing}, the
closed sets are precisely
the \emph{safety properties}~\cite{Lamport77} when derivations are
viewed as runs of a system.

\begin{example}
 Let $\zeta =\left\{\left(a \flechup{\phi_1} b \flechup{\phi_3} a\right)^n\, a\flechup{\phi_2} c ~|~n\in \mathbb N\right\}$ 
be an abstract strategy over $\caA_{lc}$. The derivation $\left(a \flechup{\phi_1} b \flechup{\phi_3} a\right)^\omega$
is a limit point of $\zeta$ and does not belong to $\zeta$. Thus, $\zeta$ is  not closed.
\end{example}

\section{Intensional Strategies}
\label{se:IntensionalStrategies}
 
In this section, strategies are considered as a way of constraining and
guiding the steps of a reduction, possibly based on what has happened in
the past.  Under this reading, at any step in a derivation, we should be
able to say whether a contemplated next step obeys the strategy $\zeta$.
This is in contrast to characterizing a set of reductions in
terms of a global property, that may depend on an entire completed
reduction.
We introduce first strategies that do not take into account the
  past; although these memoryless strategies allow us to generate a
  significant number of classical abstract strategies used in term
  rewriting, they are less powerful when it comes to generate
  strategies ubiquitous in game theory. We introduce in what follows
  these two classes of strategies and will state formally in
  Section~\ref{se:express} their expressive power.

\subsection{Memoryless strategies}

Let us first consider in this section a class of strategies that
chooses the next step only regarding the current object (or state).
We follow an established convention~\cite{dagstuhl2001} and call these
strategies \emph{memoryless}.  The following definition formalizes the
choice of the next step using a partial function on objects.

\begin{definition}[Memoryless intensional strategy]
  A \emph{memoryless intensional strategy} over an \ars\ ${\aos}$ is a partial
  function $\nzeta$ from $\ca O$ to $2^{\reductions}$ such that for every
  object $a$, $\nzeta(a) \subseteq \{\pi \mid \pi\in\Gamma, Dom(\pi)=a
  \}$.
\end{definition}

In this definition, $\pi\in\Gamma$ denotes a
reduction step or equivalently a derivation of length $1$. 

A memoryless intensional strategy naturally generates an abstract strategy, as follows.

\begin{definition}[Extension of a memoryless intensional strategy] 
  \label{def:ext-of-int}
  ~~ Let $\nzeta$ be a memoryless intensional strategy over an \ars\
  ${\aos}$.  The \emph{extension} of $\nzeta$ is the abstract
  strategy $\ezeta$ consisting of the following set of derivations:
  \[
  \pi = ((a_i,\phi_i, a_{i+1}))_{i \in\Im} \in \ezeta
  \qquad \text{ iff } \qquad
  \forall j\in\Im, \quad 
  (a_j , \phi_{j} , a_{j+1}) \in \nzeta(a_{j})
  \]

We will sometimes say that the intensional strategy $\nzeta$
  \emph{generates} the abstract strategy $\ezeta$.
\end{definition}

This  extension may obviously contain infinite derivations; in
such a case it also contains all the finite derivations that are
prefixes of the infinite ones. Indeed, it is easy to see from
Definition~\ref{def:ext-of-int} that the extension of an intensional
strategy is closed under taking prefixes.  We show next that the set of
finite derivations generated by an intensional strategy $\nzeta$ can be
constructed inductively from $\nzeta$.

\begin{definition}[Finite support of an abstract strategy] 
 \label{def:support}
Let us call the \emph{finite support} of (any) strategy $\zeta$  the set of
\emph{finite} derivations in $\zeta$ and denote it $\fzeta$.
\end{definition}

\begin{proposition} 
 \label{def:pre-ext-of-int}
Given a memoryless intensional strategy $\nzeta$ over an \ars\ 
of the form ${\aos}$,
the finite support of its extension 
is an  abstract strategy denoted $\szeta$ and inductively defined as follows:
\begin{itemize}
\item $\bigcup_{a\in Dom(\nzeta)}\nzeta(a) \subseteq \szeta$
\item $\forall \pi \in \szeta$ and $\forall\pi' \in \nzeta(Im(\pi)),~\pi \pi' \in \szeta$
\end{itemize} 
\end{proposition}

\begin{proof}
{
Clearly, the derivations that are computed by this inductive definition are in $\ezeta$ and are finite.
Conversely, by induction on the length of the derivations, any finite derivation in $\ezeta$ is built by one 
of the two inductive cases.
}
\end{proof}

Notice that an intensional strategy $\nzeta$ over an \ars\ ${\aos}$
induces a sub-ARS $\ca B=(\ca O,\ca L, \bigcup_{a\in
  Dom(\nzeta)}\nzeta(a))$ of $\ca A$ and thus, such that $\szeta =
\derivations{\ca B}$ 
and  $\ezeta = \allDerivInf{\ca B}$. Note also that $\ezeta$ is the smallest 
closed strategy containing $\szeta$. 

\begin{example}
  Let us consider the \ars\ ${\cal A}_{lc}$
 of Example~\ref{ex:simples} and define the following strategies:

  \begin{itemize}
  \item The intensional strategy $\nzeta_u$ defined on all objects
    in $\ca O$ such that for any object $a\in{\ca O}$, $\nzeta_u(a) =
    \{\pi \mid \pi\in\Gamma, Dom(\pi)=a \}$ obviously generates the
    Universal strategy $\zeta_u$ (of
    Example~\ref{ex:strategies}). Moreover
    $\sext{\nzeta_u} = \allDeriv{{\ca A}_{lc}}$.

  \item The intensional strategy $\nzeta_f$ defined on no object
    in $\ca O$ generates the Fail strategy $\zeta_f$ (of
    Example~\ref{ex:strategies}). 

  \item
    Given an \ars\ ${\cal A}$, let us consider an order $<$ on the
    labels of ${\cal A}$ and a function ``$max$'' that computes the
    maximal element(s) of a set (the result is a singleton if the order
    is total). The intensional strategy $\nzeta_{gm}$ such that
    $\nzeta_{gm}(a) = \{\pi: a \flechup{\phi} b \mid \phi=max(\{\phi'
    \mid a \flechup{\phi'} b \in \reductions \})\}$ generates a
    ``$Greatmost$'' abstract strategy $\zeta_{gm}$ that, for each of its
    derivations, chooses each time one of the steps with the greatest
    ``weight'' specified by the label.
    \begin{itemize}
    \item If we consider the \ars\ ${\cal A}_{lc}$ with the order
      $\phi_1<\phi_2<\phi_3<\phi_4$, then
     $\ext{\nzeta_{gm}}=\sext{\nzeta_{gm}} =\left\{a\flechup{\phi_2} c~;~ b\flechup{\phi_4}
       d\right\}$.

   \item If we consider the \ars\ ${\cal A}_{lc}$ with the order
      $\phi_1>\phi_2>\phi_3>\phi_4$, then
     $\ext{\nzeta_{gm}}=\left\{\left(a\flechup{\phi_1 \phi_3} a\right)^n a
        \flechup{\phi_1} b~;~ \left(a\flechup{\phi_1 \phi_3}
          a\right)^\omega
	~;~\left(b\flechup{\phi_3 \phi_1} b\right)^n b \flechup{\phi_3}
        a ~;~
	\left(b\flechup{\phi_3 \phi_1} b\right)^\omega~|~n\geq 0\right\}
      $
and\\
      $\sext{\nzeta_{gm}} =\left\{\left(a\flechup{\phi_1 \phi_3} a\right)^n a
        \flechup{\phi_1} b 
        ;
        \left(a\flechup{\phi_1 \phi_3} a\right)^m
        ;
        \left(b\flechup{\phi_3 \phi_1} b\right)^n b \flechup{\phi_3} a 
        ;
        \left(b\flechup{\phi_3 \phi_1} b\right)^m
        ~|~ m>n\geq 0 \right\}
      $
    \end{itemize}
  \end{itemize} If the objects of the \ars\ are terms and the rewriting
  steps are labeled by the redex position, then we can use the prefix
  order on positions and the intensional strategy generates in this case
  the classical \emph{innermost} strategy. When a lexicographic order is
  used, the classical \emph{rightmost-innermost} strategy is obtained.
\end{example}

\subsection{Intensional strategies with memory}

The previous definition of memoryless intensional strategies cannot take into
account the past derivation steps to decide the next possible ones.  For
that, the history of a derivation has to be memorized and available at
each step.  In order to define intensional strategies with memory,
called simply intensional strategies in what follows, let
us first introduce the notion of traced-object where each object
memorizes how it has been reached.

\begin{definition}[Traced-object]
  Given a countable set of objects $\ca O$ and a countable set of labels
  $\ca L$ mutually disjoint,
  a \emph{traced-object} is a pair $\cx{\alpha} a$ where $\alpha$ is a
  sequence of elements of $\ca O \times \ca L$ called \emph{trace} or
  \emph{history} .
\end{definition}

In this definition, we implicitly define a monoid $\big((\ca O \times
\ca L)^*,\odot\big)$ generated by $(\ca O \times \ca L)$ and whose
neutral element is denoted by $\vide$.

\begin{definition}[Traced object compatible with an ARS]
  A traced-object $\cx{\alpha} a$ is \emph{compatible} with $\aos$ iff
  $\alpha = \left((a_i,\phi_i)\right)_{i\in\Im}$ for any right-open
  interval $\Im\subseteq \mathbb N$ starting from $0$ and $a =
  a_{n}$ and for all $i\in\Im$, $(a_i,\phi_i,a_{i+1})\in\Gamma$.
  In such a case, we denote by $\sema{\alpha}$ the derivation
  $\big((a_i,\phi_i,a_{i+1})\big)_{i\in\Im}$ and by $\cobj$ the set of
  traced objects compatible with $\ca A$.
 Moreover, we define an equivalence relation $\sim$ over $\cobj$ as
  follows: $\cx{\alpha} a \sim \cx{\alpha'} a'$ iff $a = a'$. We
  naturally have $\cobj /\! \sim\,= \ca O$.
\end{definition}

We can now refine the definition of intensional strategies taking the
history of objects into account.

\begin{definition}[Intensional strategy (with memory)]
  An \emph{intensional strategy} over an \ars\ ${\aos}$ is a partial
  function $\nzeta$ from $\cobj$ to $2^{\reductions}$ such that for
  every traced object $\cx{\alpha}a$, $\nzeta(\cx{\alpha}a) \subseteq
  \{\pi\in\reductions \mid Dom(\pi)=a \}$.
\end{definition}

As for memoryless intensional strategies, an intensional strategy naturally
generates an abstract strategy, as follows.

\begin{definition}[Extension of an intensional strategy] \label{ext-of-int-mem}
  Let $\nzeta$ be an intensional strategy over an \ars\
  ${\aos}$.  The \emph{extension} of $\nzeta$ is the abstract
  strategy $\ezeta$ consisting of the following set of derivations:
 \[
  \pi = ((a_i,\phi_i, a_{i+1}))_{i\in\Im} \in \ezeta
  \qquad \text{ iff } \qquad
  \forall j\in\Im, \quad 
  (a_j , \phi_{j} , a_{j+1}) \in \nzeta(\cx{\alpha}a_{j})
 \]
  where $\alpha = \left((a_i,\phi_i)\right)_{i\in\Im}$. \quad 
\end{definition}

As before, we can inductively define the finite support of this extension 
as an abstract strategy  $\szeta$ 
containing all finite derivations of $\ezeta$.

\begin{proposition} 
 \label{def:pre-ext-of-int-with-mem}
Given an intensional strategy with memory $\nzeta$ over an
\ars\  of the form ${\aos}$,
the finite support of its extension 
is an  abstract strategy denoted $\szeta$ inductively defined as follows:
\begin{itemize}
 \item $\forall \cx\vide a\in\cobj, \nzeta(\cx\vide a) \subseteq\szeta $,
 \item $\forall \alpha$ \textit{s.t.} $\pi = \sema{\alpha} \in \szeta $ and
$\pi' \in \nzeta \left(\cx{\alpha}Im(\pi)\right),~\pi \pi' \in \szeta $
\end{itemize}
\end{proposition}

\begin{proof}
{
Similar as in Proposition~\ref{def:pre-ext-of-int}.
}
\end{proof}

\begin{example} \label{with-mem-examples}
  The following examples of strategies cannot be expressed without the
  knowledge of the history and illustrate the interest of traced
  objects.
  \begin{itemize}
  \item
    The intensional strategy that restricts the derivations to be of
    bounded length $k$ makes use of the size of the trace $\alpha$,
    denoted $|\alpha|$:
    $$\nzeta_{ltk}(\cx{\alpha}a) =  \{\pi \mid \pi\in\Gamma, Dom(\pi)=a, |\alpha| < k-1\}$$
  \item If we assume that the reduction steps are colored, for
    instance, in white or black via their labels, then the following
    intensional strategy  generates a strategy whose
  reductions alternate white and black steps:
   $$\nzeta_{WB}(\cx{\left((a_i,\phi_i)\right)_{0\leq i \leq n}}a) 
    = \{\pi: a \flechup{neg(\phi_n)} b \mid \pi\in\Gamma\}$$ with
    $neg(white)=black$ and $neg(black)=white$.  Once again, the
    knowledge of (the color) of the previous step is essential for
    choosing the current one.

  \item The strategy that alternates reductions from a set (of steps)
    $\Gamma_1$ with reductions from a set $\Gamma_2$ can be
    generated by the following intensional strategy:
    \[
    \begin{array}{rcll}
    \nzeta_{\Gamma_1;\Gamma_2}(\cx\vide a) &=&  \{\pi_1 \mid \pi_1 \in \Gamma_1, Dom(\pi_1)=a\} \\
    \nzeta_{\Gamma_1;\Gamma_2}(\cx{\alpha' \odot (u,\phi')} a) &=&    \{\pi_1 \mid \pi_1 \in \Gamma_1, Dom(\pi_1)=a\}&
    \mbox{if $u \flechup{\phi'} a \in \Gamma_2$ }\\
    \nzeta_{\Gamma_1;\Gamma_2}(\cx{\alpha' \odot (u,\phi')} a) &=&    \{\pi_2 \mid \pi_2 \in \Gamma_2, Dom(\pi_2)=a\}&
    \mbox{if $u \flechup{\phi'} a \in \Gamma_1$ }\\
  \end{array}
  \]
  As a concrete example, let $\Gamma_1=\{ a \to b, b \to c \}$ and let
  $\Gamma_2=\{b \to b\}$ (labels are omitted). Then, given the
  reduction $a \to b $, $\nzeta_{\Gamma_1;\Gamma_2}$ yields $\{b\to
  b\}$ as next step, while given the reduction $a \to b \to b$,
  $\nzeta_{\Gamma_1;\Gamma_2}$ yields $\{b\to c\}$ as next step. So
  $\nzeta_{\Gamma_1;\Gamma_2}$ is not memoryless: if all we know is the
  last element of a
  history (for example, $b$) we cannot determine the next step(s).

\item Some standard term-rewriting strategies such as \emph{parallel
    outermost} (when all outermost redexes must be contracted) or
  \emph{Gross-Knuth reduction} are not memoryless when viewed as
  strategies over the reduction system whose steps are single-step
  rewrites.
 \end{itemize}
\end{example}

\section{Expressiveness of intensional strategies}
\label{se:express}

Not every abstract strategy arises as the extension of an intensional
strategy. In this section we give a characterization of such abstract
strategies. 

\subsection{Which abstract strategies can be described by intensional strategies ?}

As a simple example, consider a single derivation $\pi$, and let $\zeta$
be the set of all prefixes of $\pi$ (including $\pi$ itself of course).   Then $\zeta$ 
is the extension of the intensional strategy that maps each
finite prefix of $\pi$ to its next step, if there is one, and otherwise
to the empty set.   Similarly, any $\zeta$ consisting of the
prefix-closures of  a finite set of
derivations is the extension of some intensional $\nzeta$.

On the other hand, 
the following example is instructive.
\begin{example} \label{eventual} 
  Let $\mathcal{A}$ be the \ars consisting of two objects $a$ and $b$ and
  the steps $a \to a$ and $a \to b$ (the labels do not matter):
  	$$
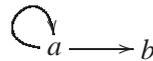
\xymatrix{
		a  \ar[r]^{}\ar@(l,u)[]_{} &  b
	}
$$
Let
  $\zeta$ be the set of all the reductions which eventually fire the
  rule $a \to b$.  Then there is no intensional strategy \nzeta such
  that $\sem{\nzeta} = \zeta.$ Suppose on the contrary that we had an
  intensional strategy determining $\zeta$.  Now ask: at a typical stage
  $a \to a \to \dots \to a$ of a reduction, does the step of firing the
  rule $a \to a$ obey this strategy?  Clearly the answer cannot be
  ``no'' at any stage since we would prevent ourselves from going
  further and firing the rule $a \to b$ eventually.  On the other hand,
  if the answer is ``yes'' at every stage, then the infinite reduction $a
  \to a \to a \to a \to \dots$ obeys our strategy at each step.  But
  this reduction is not in $\zeta$!
\end{example}

Example~\ref{eventual} shows that not every set of reductions can be
captured by an intensional notion of strategy.  Note that this is not
a result about computable strategies, or memoryless strategies, or
deterministic strategies.  And the set $\zeta$ is a perfectly
reasonable set of derivations: it can be defined by the rational
expression $ (a \to a)^* (a\to b).$ But there is no \emph{function
  on traced objects} generating precisely the set of derivations in
question.

The intuition behind this example is this: for a given intensional strategy $\nzeta$,
if a derivation $\pi$ \emph{fails} to be in $\sem{\nzeta}$, then
there is a particular step of $\pi$  which fails to obey \nzeta.  This is
the essential aspect of strategies that we suggest distinguishes them
from other ways of defining sets of reductions: their local, finitary
character.

As a preliminary to the following result, we show that 
the family of sets of reductions determined by strategies is closed
under arbitrary intersection.  Indeed, the following stronger
observation is easy
to verify.

\begin{lemma}
  \label{ints}
 Let $\Sigma = \{\nzeta_i \mid i \in \Im \subseteq \mathbb N \}$ be any set of
intensional  strategies and  $\nzeta$  the pointwise intersection of the
  $\nzeta_i$, that is, %
  $\nzeta(\cx{\alpha}a
) =
  \bigcap \{\nzeta_i(\cx{\alpha}a
) \mid i \in \Im \}$. Then $\sem{\nzeta} = \bigcap
  \{\sem{\nzeta_i} \mid i \in \Im\}$.  
\end{lemma}
\begin{proof} {An easy calculation based on definition~\ref{ext-of-int-mem}.}
\end{proof}

\begin{proposition} \label{char}
  Let $\zeta$ be a set of (non-empty) derivations. There exists an intensional strategy $\nzeta$ with 
  $\ezeta = \zeta$ iff $\zeta$ is a closed set.
\end{proposition}

\begin{proof}
  { Let $\nzeta$ be an intensional  strategy.  To show that $\ezeta$ is
   closed, we must show that if $\pi$ is a derivation that is not in
   $\ezeta$, then there is a finite prefix $\pi_0$ of $\pi$ such
    that every extension $\pi'_0$ of $\pi_0$ fails to be in   $\ezeta$.
   Write $\pi$ as the sequence of steps $s_1, s_2,
    \dots$ where each $s_i$ is an element $(a_i, \phi_i, a_{i+1}) \in
      \Gamma$.  If $\pi$ is not in  $\ezeta$,   then for some $i$, the
    $i$th step $s_{i} = (a_i, \phi_i, a_{i+1}) \notin \ezeta(\cx{\alpha}a_{i})$
   where $\alpha = \left((a_j,\phi_j)\right)_{0\leq j < i}$.
   We can take $\pi_0$ to be  $\sema{\alpha}$, the derivation composed of the $(i-1)$
    first steps of $\pi$.
   Here we use the fact that if 
$\nzeta(\cx{\alpha}a_{i}) = \emptyset$ then $\nzeta(\cx{\alpha'}a') = \emptyset$ for 
all $\alpha'$ such that $\sema{\alpha'} = \pi' = \pi_0 \pi_1$ for some $\pi_1$.

    For the converse, suppose that $\zeta$ is a closed set of reductions.
    Then $\zeta$ is the intersection of the set of the complements of the
    basic open sets disjoint from $\zeta$.  By Lemma~\ref{ints}
    it suffices to show that
    the complement of any basic open set is defined by an intensional strategy.  So fix
    a finite $\pi_0$; we need to construct an intensional strategy determining the set of
    those $\pi'$ which do \emph{not} extend $\pi_0$.  Letting $\pi_0 = 
    s_1, s_2, \dots, s_n $ we may simply define the strategy $\zeta$
    to return the empty set on all $\pi'$ extending $\pi_0$ and return all
    possible next moves on all other inputs. 
  }
\end{proof}

\begin{example}
 \label{se:example}

Consider the following \ars $\mathcal{A}$,  modeling a simple
intersection with two traffic signals. 
There are two directions (say,
the north-south direction and the east-west direction).
One traffic signal controls the north-south  direction, another  controls
the east-west direction: each 
signal can be red or green.
With each direction is associated a queue of cars waiting to cross:  we
can model these as natural numbers.  (We can also bound the number of
cars allowed in the model in order to obtain a finite \ars; this does
not affect the observations made in this example.)

So an object of $\mathcal{A}$  is a quadruple $[q_1, l_1, q_2, l_2]$ where
$q_1$ and $q_2$ are natural numbers and 
$l_1$ and $l_2$ can each take the value 0 (for ``red'') or 1 (for ``green'').

The steps in $\mathcal{A}$ correspond to the fact that at each time unit
\begin{itemize}
\item  cars can arrive at queue 1 or queue 2

\item some signal can change its value 

\item a car may cross the intersection if its signal is green
\end{itemize}
This leads to the following set of labels:
$\mathcal{L} = \{
\carone, \cartwo, \signalone, \signaltwo, \crossone, \crosstwo
\}$.
The steps of $\mathcal{A}$ may be defined by the following schematic reduction steps

\begin{align*}
  &  [q_1, l_1, q_2, l_2 ] \stackrel{\carone}{\longrightarrow} [(q_1+1),  l_1, q_2,   l_2] \qquad
  \;[q_1, l_1, q_2, l_2] \stackrel{\cartwo}{\longrightarrow} [q_1,  l_1,  (q_2+1),   l_2] \\ 
  & [q_1, l_1, q_2, l_2] \stackrel{\signalone}{\longrightarrow} [q_1,  (1- l_1),  q_2,  l_2] \qquad
  [q_1, l_1, q_2, l_2] \stackrel{\signaltwo}{\longrightarrow} [q_1,  l_1,  q_2,  (1- l_2)] \\ 
  & [q_1, 1, q_2, l_2] \stackrel{\crossone}{\longrightarrow} [(q_1 - 1),  1,  q_2,   l_2] \qquad
  \; [q_1, l_1, q_2, 1] \stackrel{\crosstwo}{\longrightarrow} [q_1,  l_1,  (q_2 -1) ,  1] \\ 
\end{align*}
Note that $\mathcal{A}$ %
models the fact that cars may arrive at and cross the intersection in
arbitrary patterns,  and it reflects the constraint that the cars obey the traffic signals.
But $\mathcal{A}$ does not, as an \ars, attempt to model an intelligent protocol for
scheduling the traffic signals: it is an arena for developing and analyzing such
a protocol.   

The system admits some states that are intuitively undesirable, for example any
state where both signals are green.  There are also some \emph{derivations} that are
undesirable, for example a derivation in which some car is left waiting at an intersection
forever, such as the ``unfair''
\begin{equation}\label{derone}
 [1, 0, 1, 1] \stackrel{\crosstwo}{\longrightarrow} 
 [1, 0, 0, 1] \stackrel{\cartwo}{\longrightarrow} 
 [1, 0, 1, 1] \stackrel{\crosstwo}{\longrightarrow} 
 [1, 0, 0, 1] \dots
\end{equation}
 An algorithm to manage the traffic signals is precisely an intensional
strategy for the \ars $\mathcal{A}$, and   the extension of such a strategy is the set of
behaviors of the system that the strategy enforces.

Here are some results about strategies in this system;    the second and
fourth items are results that follow easily from the characterization theorem.
\begin{itemize}
\item There are intensional strategies that ensure that the signals are
  never both green simultaneously (this is easy).

\item More interestingly, there is an intensional strategy
  $\nzeta$ such that $\sem{\nzeta}$ is precisely the set of those
  derivations such that the signals are never both green simultaneously.
  This is a highly non-deterministic strategy, that permits \emph{any}
  behavior as long as it does not permit simultaneous greens.

\item There are intensional strategies to ensure that any car arriving
  at an intersection is eventually allowed through.

\item But there is \emph{no} intensional strategy $\nzeta$ whose extension 
  $\sem{\nzeta}$ is precisely the
  set of all those paths in which any car arriving at an
  intersection is eventually allowed through. 
\end{itemize}
The second fact above follows from the observation that the set of
derivations in which the signals are never both green simultaneously is
a closed set (this is an easy consequence of Definition~\ref{closed}).
The fourth fact above follows from the observation that the set of
``fair'' derivations, in which no car is forever denied access to the
intersection, is not a closed set.  To see this, note that the unfair
derivation above is a limit point of the set of fair derivations, since
every finite prefix of this derivation can be extended to one in which
the car waiting (and all subsequent cars) crosses the intersection.
Since the derivation itself is not fair, we see that the set of fair
derivations does not contain all of its limit points, and so is not
closed.
 \end{example}

\subsection{Closure properties for intensional strategies}
\label{sec:closure-under-unions}

In light of the importance of designing languages for expressing complex
strategies, a natural question to ask is: what are the closure
properties enjoyed by (the extensions of) intensional strategies?
We observed in Lemma~\ref{ints} that
this class is closed under intersection.  Indeed,
 by taking the pointwise intersection of a family 
$\{\nzeta_i \mid i \in I \}$ of intensional  strategies, 
the extensional strategy generated is the intersection of the extensions
of the $\nzeta_i$.
However this pointwise construction fails for arbitrary union. Even when
we restrict to finite unions the situation is subtle:
if $\nzeta_1$ and $\nzeta_2$ are intensional strategies 
generating $\ext{\nzeta_1}$ and $\ext{\nzeta_2}$ respectively, 
and  if we write $\nzeta_1 \cup \nzeta_2$ for the intensional strategy that is
the pointwise union of 
 $\nzeta_1$ and $\nzeta_2$, then
$\nzeta_1 \cup \nzeta_2$ will \emph{not,}  in
general, generate 
 $\ext{\nzeta_1} \cup \ext{\nzeta_2}$,   as the
next example demonstrates.

\begin{example}
 Given the \ars with objects $\{a, b_1, b_2\}$ and
  reduction steps 
  $(a, \phi_1, b_1)$, $(a, \phi_2, b_2)$, $(b_1, \beta_1, a)$,
  $(b_2, \beta_2, a)$,
  let $\nzeta_1$ be the (memoryless) intensional strategy %
  \[
  a \mapsto \{(a, \phi_1, b_1)\}, \quad b_1 \mapsto \{(b_1, \beta_1,  a) \}
  \]
  and  let $\nzeta_2$ be the intensional strategy %
  \[
  a \mapsto \{(a, \phi_2, b_2)\}, \quad b_2 \mapsto \{ (b_2, \beta_2, a) .
    \]
Clearly $\ext{\nzeta_1}$ 
is the set of derivations that loop between $a$
and $b_1$, and similarly for  $\ext{\nzeta_2}$.
If we now construct the intensional strategy $\nzeta_1 \cup
    \nzeta_2$ by taking the pointwise union of $\nzeta_1$ and
    $\nzeta_2$; thus 
    \[
    \nzeta_1 \cup  \nzeta_2 \quad = \quad 
    a \mapsto \{(a, \phi_1, b_1), (a, \phi_2, b_2) \}, b_1 \mapsto \{(b_1, \beta_1,
    a)\}, b_2 \mapsto \{ (b_2, \beta_2, a) 
    \]
then clearly $\ext{\nzeta_1 \cup  \nzeta_2}$
is the set of all derivations.
Thus
\[
\ext{   \nzeta_1 \cup  \nzeta_2  } \neq  \ext{ \nzeta_1} \cup  \ext{ \nzeta_2 }
\]
That is, the pointwise union of intensional strategies does not give
rise to the union of the corresponding extensions.
\end{example}

Nevertheless,  $\sem{\nzeta_1} \cup \sem{\nzeta_2}$ is indeed
generated by an intensional strategy (with memory): at the first step at object $a$
we non-deterministically move to $b_1$ or to $b_2$ and at subsequent
steps we always make the same choice.
It is not an accident that we can generate $\sem{\nzeta_1} \cup
\sem{\nzeta_2}$ intensionally, as we see next.

\begin{proposition}
 The class of abstract strategies that are extensions of intensional strategies is
 closed under arbitrary intersection and finite unions; it is not
  closed under complement.
\end{proposition}
\begin{proof}
  {
    These are immediate consequences of Proposition~\ref{char} and basic facts
    about topological spaces.
 }
\end{proof}

The fact that unions of extensions of intensional
strategies are intensionally generated--even though the naive
``pointwise union'' construction fails--
is a nice application of the
topological perspective.

\subsection{Other classes of abstract strategies that can be intensionally described}

When restricting to abstract strategies consisting of a
  potentially infinite number of finite derivations, the existence of
  a corresponding intensional strategy depends on some classical
  properties of the original strategy. The proof of  existence of
  a corresponding intensional strategy under appropriate
  assumptions explicits the way such a strategy can be built.

\begin{proposition}
Given an abstract strategy $\zeta$ over $\aos$ consisting only
   of finite derivations, there is a memoryless
    intensional strategy $\nzeta$ over $\ca A$ such that  $\szeta = \zeta$ iff $\zeta$ is
    factor-closed and closed under composition.
\end{proposition}

\begin{proof}
  { Obviously, the abstract strategy $\szeta$ built from an
    intensional strategy $\nzeta$ is factor-closed and closed by
    composition.  Conversely, if $\zeta$ is factor-closed and closed
    by composition, $\nzeta$ can be defined as follows: for all
    object $a$, %
    $\nzeta(a) = \{a \flechup{\phi} b \mid a
      \flechup{\phi} b \in \zeta \cap \reductions . \} $ }
\end{proof}

The existence of memoryless intensional strategy obviously
  implies the existence of an intensional strategy  (with memory). The class
  of abstract strategies that satisfy the above conditions is already
  quite important, especially when considering term rewriting
  strategies but, as we have seen in
  Section~\ref{se:IntensionalStrategies}, there are strategies that do
  not fit these constraints. 

For example, if we consider again Example~\ref{se:example} over a finite interval of time, we can find	 
   a memoryless intensional strategy that corresponds to the set of derivations 	 
   such that the signals are never both green simultaneously. On the other hand 
   there is an intensional strategy, but not a memoryless one,
   for the set of derivations such that a car	 
  waits at most $n$ turns before crossing the intersection.   This
  follows from the following proposition.

\begin{proposition}
Given an  abstract strategy $\zeta$ over $\aos$  consisting only
   of finite derivations, there is an
    intensional strategy $\nzeta$ over $\ca A$ such that  $\szeta = \zeta$ iff $\zeta$ is
    prefix-closed.
\end{proposition}

\begin{proof}{
    Let $\nzeta$ be an intensional strategy over $\aos$ and
    $\pi\in\szeta$; two cases are possible:
    \begin{itemize}
    \item $|\,\pi\,| \leq 1$, then $\pi$ has no prefix.
    \item $|\,\pi\,| > 1$, then there exists $\alpha$ \textit{s.t.}
      $\pi_{pref}=\sema{\alpha} \in \szeta $ and
      $\pi' \in \nzeta \left(\cx{\alpha}Im(\pi)\right),~\pi = \pi_{pref}
      \pi'$. By applying the same reasoning over $\pi_{pref}$, we obtain
      that all prefixes of $\pi$ are in $\szeta$.
    \end{itemize} Conversely, let us consider a prefix-closed abstract
    strategy $\zeta$. The intensional strategy $\nzeta$ defined as
    follows:
    \begin{itemize}
    \item $\forall a\in\ca O, \nzeta(\cx{\vide}a) = \{\pi\in\zeta \cap
      \reductions \mid Dom(\pi)=a\}$,
    \item $\forall \pi\in\zeta$, and $\cx{\alpha} a \in \cobj$
      \textit{s.t.} $\sema{\alpha} = \pi$, $\nzeta(\cx\alpha a) = \{\pi'
      \in \reductions \mid \pi \pi' \in \zeta\}$
    \end{itemize} is such that $\szeta = \zeta$.
  }
\end{proof}

\section{Logical intensional strategies}
\label{sec:log-int-strat}

Instead of defining an intensional strategy $\nzeta$ by a function, we
can consider using a logical approach and  associating to $\nzeta$ a
\emph{characteristic property} denoted by ${\ca P}_{\nzeta}$ such
that:
\begin{center}
${\ca P}_{\nzeta}(\cx{\alpha} a, \phi)$ is true iff 
$\exists b$ such that $(a,\phi,b)\in \nzeta(\cx\alpha a)$
\end{center}
Thus, $\szeta$  is the following prefix closed abstract strategy:
$$
\left\{
\sema{\alpha\odot(a,\phi)} \mid 
{\ca P}_{\nzeta}(\cx{\alpha} a, \phi) \textit{~and~}
(\alpha = \vide \vee  \sema\alpha\in\szeta) \textit{~and~}
\exists b, a \flechup{\phi} b \in \reductions 
\right\}
$$

\begin{example}
  Let us show how some of the previous examples and some variants can be
  easily expressed with a characteristic property.

  \begin{itemize}
  \item For the $Universal$ strategy $\nzeta_u$ over an ARS  ${\aos}$, ${\cal P}_{u}(\cx{\alpha} a, \phi) =
    \top$, where $\top$ denotes the true Boolean value.

  \item For the $Fail$ strategy $\nzeta_f$ over an ARS  ${\aos}$, ${\cal P}_{f}(\cx{\alpha} a,\phi) =
    \bot$, where $\bot$ denotes the false Boolean value.

  \item For the $Greatmost$ strategy $\nzeta_{gm}$ over an ARS ${\aos}$
    with an order $<$ on the labels, ${\cal P}_{gm}(\cx{\alpha} a,\phi) =
    \forall (\phi',b),a \flechup{\phi'} b \in \Gamma \implique \phi \not< \phi'$.

   \item For the strategy $\nzeta_{ltk}$ that selects the set of
     derivations of length at most $k$, 
     ${\cal P}_{ltk}(\cx\alpha a,\phi) =  |\alpha| < k$.

   \item For the strategy $\nzeta_{R_1;R_2}$ that alternates
     reductions with labels from $R_1\in \ca L$ with reductions with labels from  $R_2\in \ca L$,
     $${\cal P}_{R_1;R_2}(\cx\alpha a,\phi) =  \left\{
     \begin{array}{lll}
       &\alpha = \vide ~~\implique~~ & \left( \phi \in R_1 \vee \phi \in R_2 \right)\\
        \wedge&  \alpha = \alpha' \odot (u,\phi')  ~~\implique~~ & \big((\phi' \in R_1 \implique \phi \in R_2)
           \vee 	   (\phi' \in R_2 \implique \phi \in R_1)\big)
 	   \end{array}\right\}
     $$
   \end{itemize}
\end{example}

Indeed, using a logical property instead of a fonction is rather a matter of choice or can be related 
to the properties of strategies we want to study, but this does not bring more expressivity.
As previously said, intensional strategies generate only closed sets of derivations and thus 
always contain all prefixes of derivations.  This prevents us from computing extensional strategies
that look straightforward like the one in the next example.

\begin{example}\label{ex:logical}
  We consider again the \ars\ ${\ca A}_{lc}$ and a strategy reduced to
  only one derivation $\zeta = \{ a \flechup{\phi_1} b
  \flechup{\phi_3} a \flechup{\phi_2} c \}$.  $\zeta$ cannot be
  computed by an intensional strategy $\nzeta$ built as before since
  its extension would contain  too many derivations, namely all prefixes of the
  derivation in $\zeta$.
\end{example}

In order to avoid this constraint, we
characterize \emph{accepted derivations} belonging to $\ezeta$ by
defining a property over $\cobj$ called \emph{accepting states} and denoted by $\ca F_\nzeta$.
This leads to the following extended definition of an intensional strategy:

\begin{definition}
  We call \emph{logical intensional strategy} over $\aos$ any pair
$(\nzeta, \ca F_\nzeta)$ where $\nzeta$ is an intensional strategy with memory 
and $\ca F_\nzeta \subseteq \cobj$.
\end{definition}

A logical intensional strategy $(\nzeta, \ca F_\nzeta)$ generates the abstract strategy
$\{ \sema{\alpha} \in \ezeta ~|~ \cx\alpha Im(\alpha) \in
{\ca  F}_\nzeta\}$.

In practice, it may be useful to describe the set ${\ca F}_\nzeta$ by its characteristic function $F_{\nzeta}$ and to test 
whether $F_{\nzeta}(\cx\alpha Im(\alpha))= true$.

Let us illustrate on  simple examples the expressive power gained in this extended definition.

\begin{example}
Coming back to Example~\ref{ex:logical}, we can now characterize the only derivation of interest by 
simply stating that ${\ca F}_{\nzeta}$ contains
only the compatible traced-object $\cx{(a,\phi_1)\odot(b,\phi_3)\odot(a,\phi_2)} c$.

In order to define a strategy that selects derivations of length
greater than $k$, we cannot  proceed as for defining $\nzeta_{ltk}$
since the strategy is not prefix closed, but we can characterize
accepting states, namely those reached in more than $k$ steps:
$\ca F_{gtk} =  \{\cx{\alpha}a ~|~ |\alpha| \geq k\}$. 
The situation is similar when one wants  to define  a strategy that selects derivations of length
exactly $k$.

As a last example, we can now formalize the strategy of
    Example~\ref{eventual} with the following accepting states $\ca
    F_{\nzeta}= \{\cx{\alpha}b ~|~ \exists n \in \mathbb N , \sema{\alpha}: a \flechup{n} a
    \flechup{} b\}$.  
\end{example}

\section{Conclusion}
\label{Conclusion}

We have proposed and discussed in this paper different definitions of
strategies stressing different aspects: clearly the notion of abstract
strategy is appropriate to explore semantic properties, while
intensional strategies are more adequate for operational purposes.  We
have tried to show how these two views do not exclude each other but
rather may be very complementary.

In order to express interesting strategies in an operational way, we
have introduced traced objects that memorize their history.  There is an
interesting analogy with languages which is worth exploring: derivations
(histories) are words built on 3-tuples $(a,\phi,b)$; the strategy (seen
as a set of derivations) is the language (set of words) to recognize;
the characteristic property of the strategy is the way to decide whether
a word belongs to the language.  Based on this analogy, it would be
interesting to see if it is possible to characterize classes of
recognizable and computable strategies.

Another direction we want to explore is the definition of intensional
strategies that, at a given step, can look forward in the following
intended derivation steps. Formalizing such looking-forward intensional
strategies is motivated for instance by looking-ahead mechanisms in
constraint solving.

Finally, we have only sketched the definition of intensional strategies 
with accepting conditions and this approach needs further work.

The domain of strategic reductions is yet largely unexplored and
important questions remain.  Let us mention further topics that have not
been addressed in this paper but we think interesting to explore.
\begin{itemize}

\item
An important topic that requires yet further exploration is the design of a
\emph{strategy language} whose purpose is to give syntactic means to describe strategies.
Actually, several systems based on rewriting already provide a strategy language, for instance
\elan~\cite{KirchnerKV-MIT95,BKKR-IJFCS-2001}, \stratego~\cite{Visser-Stratego-RTA-2001},
\strafunski~\cite{strafunski},
\tom~\cite{Tom-Manual-2.4} or more recently \maude~\cite{MartiOlietMeseguerVerdejo04}.
It is interesting to identify common  constructs provided in these different languages, 
and to try to classify them according to their use either to explore the structure of objects (here terms)
or to build derivations or sets of derivations.
\begin{itemize}
\item
Basic constructions are given by rewrite rules whose application corresponds to an elementary 
reduction step. Identity and failure are also present as elementary constructions.
Their semantics is given by their abstract strategy definitions.
\item
Due to the tree structure of terms, traversal strategies that give access to sub-terms
are based on two constructions $All$ and $One$ that consider immediate sub-terms 
of a given node: 
on a term $t$, $All(s)$ applies the strategy $s$ on all immediate sub-terms, while 
$One(s)$ applies the strategy $s$ on the first immediate sub-term where $s$ 
does not fail.
\item
Operations to build derivations are sequential composition $Sequence$ and choice,
that may be deterministic ($Choice$) or not ($ND-Choice$).
Another construction, $Try$, which gives an unfailing choice, is very useful and can be 
just derived  from the previous ones.
\item
With a functional view of strategies, it is natural to define recursive strategies and to define a fixpoint 
operator. This is the way to perform iteration and for example to construct the $Repeat$ operator.
\item 
Traversal strategies also use the fixpoint operator to program 
different ways to go through the term structure, as in $Innermost$ or $Outermost$ strategies.
Such traversals are typical intensional strategies, as described in this paper.
\end{itemize}
These constructions are devoted to rewriting on terms or term graphs. Indeed, extending the language to graph rewriting 
raises new challenges, such as graph traversal.
Other constructions could be interesting in more general contexts than term rewriting, especially
parallel application of strategies with indeed non-interference.

\item
Proving properties of strategies and strategic reductions has already been explored 
in the case of specific strategies. Let us mention in particular the following approaches
for specific properties:
confluence, weak and strong termination, completeness of strategic rewriting have been addressed
in~\cite{GK-TOCL09,GK-JPJ2007,GK-C-red-2006,FGK-WRLA-2002}  for several traversal strategies 
using a schematization of derivation trees
and an inductive proof  argument. Another approach is the  dependency pairs technique which has been 
adapted to prove termination of rewriting under innermost\cite{Arts-Giesl-TCS-2000} or lazy strategies\cite{GieslSST-Haskell-RTA-2006}.
Other works such as~\cite{FGK-PPDP-2003,ZantemaR-WRS08} have considered strategies 
transformation to equivalent rewrite systems that are preserving properties like termination.
However other properties, such as fairness or  loop-freeness,  have been much less studied.
In general, we may expect that the logical characterization of intensional strategies could 
help to prove derivation properties in the context of strategic
derivations.

\item We have distinguished between arbitrary intensional strategies and
  \emph{memoryless} strategies.  In the well-studied domain of games on
  finite graphs--so important in verification research---there is an
  important middle ground between these: the class of strategies
  requiring a fixed finite amount of memory.  Such strategies can be
  computed by a finite-state machine \cite{DziembowskiJW97}.  The
  question of finite-memory strategies over abstract reductions systems
  is subtle, essentially due to the fact that an \ars corresponds to a
  (solitaire) game on a typically \emph{infinite} arena.  Note that even
  a conceptually simple term-rewriting strategy such as
  parallel-outermost cannot be said to require a fixed finite amount of
  memory, because terms can have an unbounded number of
  parallel-outermost redexes (see Example~\ref{with-mem-examples}).  A
  careful treatment of the proper analogue of \emph{finite-memory}
  strategies over general abstract reduction systems is an interesting
  topic for future work.
\end{itemize}

\bibliographystyle{eptcs}

\end{document}